\documentclass[aps,showpacs]{revtex4}
\usepackage{amsmath,amsfonts,amsthm}
\usepackage{epsfig}
\usepackage{amssymb}

\newtheorem{theorem}{Theorem}
\newtheorem{remark}{Remark}
\newtheorem{lemma}{Lemma}

\begin{document}

\title{Localization phenomena in Nonlinear Schr\"odinger equations with spatially inhomogeneous nonlinearities: \\ Theory and applications to Bose-Einstein condensates}

 \author{V\'ictor M. P\'erez-Garc\'ia} 

\affiliation{Departamento de Matem\'aticas, E. T. S. I. Industriales and Instituto de Matem\'atica Aplicada a la Ciencia y la Ingenier\'{\i}a, 
Edificio Polit\'ecnico, Avenida de Camilo Jos\'e Cela, 3, Universidad de Castilla-La Mancha 13071 Ciudad Real, Spain} 

\author{Rosa Pardo}

\affiliation{Departamento de Matem\'atica Aplicada, 
Facultad de Ciencias Qu\'{\i}micas, 
Universidad Complutense,  Avda. Complutense s/n, 
28040 Madrid, Spain} 

\begin{abstract}
We study the properties of the ground state of Nonlinear Schr\"odinger Equations with spatially inhomogeneous
interactions and show that it experiences a strong localization on the spatial region where the interactions vanish. At the same time, tunneling to regions with positive values of the interactions is strongly supressed by the nonlinear interactions and as the number of particles is increased it saturates in the region of finite interaction values. The chemical potential has a cutoff value in these systems and thus takes values on a finite interval. 
The applicability of the phenomenon to Bose-Einstein condensates is discussed in detail.
\end{abstract}

\pacs{03.75.Lm, 05.45.Yv, 42.65.Tg}

\maketitle

\section{Introduction}

The experimental generation of Bose-Einstein condensates (BEC) with ultracold dilute atomic
vapors \cite{Pitaevskii} has turned out to be of exceptional importance for physics. 
The formation of a condensate occurs when the temperature is low enough and most of the atoms occupy the ground state of the system. This process is visible both in momentum and in real space due to the spatial inhomogeneities exhibited by the order parameter on a macroscopic scale because of the trapping potentials.

The properties of the ground state of trapped BECs are well known. In the mean field limit
simple analytical expressions are available  in the  Thomas-Fermi approximation \cite{Thomas-Fermi} and beyond \cite{Dalfovo1,Dalfovo2}. These approximations and direct numerical simulations describe accurately the properties of the experimentally found ground states \cite{Numerics1,Numerics2}.
 
Nonlinear interactions between atoms in a Bose-Einstein condensate are dominated by the two-body collisions that can be controlled by the so-called Feschbach resonance (FR) management \cite{FB1}. The control in time of the scattering length has been used to generate bright solitons \cite{b1,b2,b3} and induce collapse \cite{collapse} and has been the basis for theoretical proposals to obtain different types of nonlinear waves \cite{Kono1,Kono2,AB,ST,Gaspar,Ft3,KP}. Moreover, interactions can be made spatially dependent by acting on either the spatial dependence of the magnetic field or the laser intensity (in the case of optical control of FRs \cite{FB2}) which act on the Feschbach resonances. This possibility has motivated many theoretical studies on the behavior of solitons in Bose-Einstein condensates (BECs) with spatially inhomogeneous interactions \cite{V1,V2,V3,V4,V5,V6,V7,V8,V9,V10,V11,V12,V13,V14,V15}. 
 
  In this paper we describe  a novel nonlinear phenomenon related to the ground state of a BEC in the mean field limit  when the interactions are spatially inhomogeneous. We will show that when the scattering length is non-negative and vanishes on a certain spatial region, a striking localization phenomenon of the atom density occurs at the regions where the interactions vanish. By tuning appropriately the control (magnetic or optical) fields this phenomenon can be used to design regions with large particle densities and prescribed geometries. Another interesting phenomenon to be studied in this letter is the nonlinear limitation of tunneling of atoms to the regions in which the interactions are stronger and the fact that the chemical potential exhibits a cutoff value.
  
  Although it is reasonable to think that a BEC will avoid regions of large repulsive interactions and
prefer to remain in regions with smaller interactions, the localization phenomenon to be described in this paper here goes beyond what one would naively expect.
    
\section{The problem and its mathematical modelling}

\subsection{Physical system to be studied}

In this paper we will consider physical systems ruled by the nonlinear Schr\"odinger equation (NLS)
\begin{equation}
i\frac{\partial \psi}{\partial t}=-\frac{1}{2}\Delta \psi+ V(x)\psi + g(x)|\psi|^{p}\psi, \label{gp_dimensionless}
\end{equation}
in $\mathbb{R}^N$, where $p>1$ is a real parameter and $g,\ V\geq 0$ are non-negative real functions.  $V$ describes an external localized
potential acting on the system satisfying,
\begin{equation}
\label{hip:V}
V(x) \to \infty, \qquad  \mbox{as } |x| \to \infty,
\end{equation}
 and $g$ is a spatially dependent nonlinear coefficient. 
Stationary solutions of Eq. (\ref{gp_dimensionless}) are defined through 
\begin{equation}
\psi(x,t) = \phi(x) \exp\left(i \lambda t\right)
\end{equation}
which leads to 
 \begin{equation}\label{gs}
 \lambda \phi = -\frac{1}{2} \Delta \phi + V(x) \phi + g(x) |\phi|^{p} \phi.
 \end{equation}
 Of all the posible solutions of Eq. (\ref{gs}) we will be interested on the so-called ground state, which is the 
 the real, stationary positive solution of the Gross-Pitaevskii equation (\ref{gs})
 which minimizes the energy functional 
 \begin{equation}
 E(\phi) =  \int_{\mathbb{R}^3} \left[  \frac{1}{2} \left|\nabla \phi\right|^2 + V(x) \left|\phi\right|^2 + \frac{1}{2} g(x) |\phi|^{p+2}\right],
 \end{equation}
under the constraint of a fixed $ \int_{\mathbb{R}^3} |\phi|^2 d x$.

For our purpouses the only relevant property is the positivity, thus we will consider positive solutions of 
\begin{equation}
\label{eq:rd}
-\frac{1}{2}\Delta  u + V(x) u =\lambda  u - g(x)u^p ,\qquad x\in \mathbb{R}^N
\end{equation}
where $V$ satisfies \eqref{hip:V}. 

From the physical point of view Eq. (\ref{eq:rd}) is a model of different physical systems. However, as stated in the introduction, in this paper we will focus mainly on its applicability to the description of the ground state of a trapped Bose-Einstein condensate with spatially inhomogeneous interactions in the mean field limit
where physically it is both posible to control the interactions and achieve large values for the coefficient $g(x)$.  Although Eq. (\ref{gs}) is written in nondimensional units it is easy to connect this problem with the realistic quantum problem (see e.g. \cite{Pitaevskii}), considering that the coordinates $x$ and time $t$ are measured in units of  $a_0=\sqrt{\hbar/m w}$ and $1/w$, respectively, while the energies and frequencies are measured in  units of $\hbar w$ and $w$ respectively,  $w$ being  a characteristic frequency of the potential. Finally, $g(x) =4\pi a_s(x)/a_0$ is proportional to the local value of the s-wave scattering length $a_s(x)$ and the parameter $p=2$. Another relevant physical quantity, the number of particles, is directly related to the $L^2$ norm of the solution through
\begin{equation}
N = \|u\|_2^2 = \int_{\mathbb{R}^N} u^2 dx.
\end{equation}

In this paper we will study properties of the positive solution of Eq. (\ref{eq:rd}) (ground state) when the interactions vanish on a certain set of points which we will denote as $\omega$, i.e.
\begin{equation} \label{set:supp:g}
\omega  :=\{ x\in \mathbb{R}^N : g(x)=0\},
\end{equation}
Through this paper we will assume that $\omega$ 
is composed by a finite number of closed connected components $\omega _j ,$ $1 \leq j \leq J,$ which are disjoint $\omega _i \cap
\omega _j =\emptyset$ if $ i\neq  j$ and it is assumed that each component $\omega _j $ is regular enough and that
$\mathbb{R}^N\setminus \omega $ is connected for $N\geq 1.$

\subsection{The Thomas-Fermi limit}

The Thomas-Fermi approximation is used  extensively in the field of Bose-Einstein condensation to describe the solutions in the situation in which the nonlinear term is ``large" and proceeds by neglecting the kinetic energy or equivalently the term proportional to $\Delta \phi$ in Eq. (\ref{gs}) which leads to 
\begin{equation}
\phi_{TF}(x) \simeq \sqrt{\frac{\lambda-V(x)}{g(x)}}.
\end{equation}
Thus, the TF approximation provides the profile of the stationary positive solution (i.e. the ground state). In  the case of spatially homogeneous interactions $g(x) = g_0$ the so-called classical turning points for which $\lambda = V(x)$ delimit the regions in which the approximation breaks, since near to those points the amplitude is small and thus the kinetic energy cannot be neglected anymore. 

It is remarkable that, when the interactions are spatially inhomogeneous a new feature appears, which is that the amplitude of the Thomas-Fermi solution diverges on the set $\omega$. 
Although it is obvious, that in the vicinity of  $\omega$ the Thomas-Fermi approximation is not correct, this divergence is the first indication of an interesting nonlinear  phenomenon with relevant physical implications whose study is the purpouse of this paper: \emph{the tendency of positive solutions of Eq. (\ref{eq:rd}), or in physical terms the ground state of Bose-Einstein condensates with spatially inhomogeneous interactions, to localize strongly on the regions where the scattering length is close to zero provided the system is sufficiently nonlinear}.

 \section{Localization phenomena: A ``toy" example}
 \label{toy}
 
 Let us first consider the exactly solvable  one dimensional ``toy" model on the interval $ x \in [-L,L],$ given by
 \begin{subequations}
 \begin{eqnarray}
 & & \lambda u = -u'' + g(x) u^3, \\
 & &u(L) = u(-L) = 0,
 \end{eqnarray}
 \end{subequations}
which can be understood as a version of Eq. (\ref{eq:rd}) with a potential of the form
\begin{equation}
V(x)= \begin{cases} 0 &  |x|<L, \\ \infty & |x|>L.
\end{cases}
\end{equation}
 We will take the interactions as given by the equation
 \begin{equation}
 g(x) = \begin{cases} g_0 & |x| < a,\\  0 & |x| >a.\end{cases}
 \end{equation}
  This is a model for an ideal quasi-one dimensional BEC in a box.
   
  In this simple case, 
the positive, stationary solution, can be obtained analitically and is given by  
\begin{subequations}
\begin{equation}
\label{pepa}
u(x)   =   C \sin \left[\sqrt{2\lambda} \left(x-L\right)\right], a <|x| < L,
 \end{equation}
while for $|x|<a$ the solutions are given by
 \begin{equation}\label{pipa}
u(x)    = \begin{cases} \sqrt{\frac{\alpha \lambda}{g_0}} \ \text{sn}\left(x\sqrt{\lambda \alpha}+\delta; k^2\right),   & \lambda < \lambda_*\\
 \pi/\left(2\sqrt{2g_0}|a-L|\right), & \lambda = \lambda_*, \\
 \sqrt{\frac{\alpha \lambda}{g_0}} \ \text{dc}\left(x\sqrt{\lambda \alpha}; k^2\right), &  \lambda > \lambda_*
  \end{cases}
   \end{equation}
   \end{subequations}
   where $\alpha(k) = 2/(1+k^2)$, sn and dc are two of the standard Jacobi elliptic functions and $k$ is the elliptic modulus. Both the elliptic modulus and amplitude $C$ can be obtained from the matching conditions for $\phi(a)$ and $\phi'(a)$. These conditions also give $\lambda_* = \pi^2/\left[8(a-L)^2\right]$. 
   
   It is interesting to point out that there is a cutoff value for the eigenvalue,  
    $\lambda_c$ beyond which stationary solutions do not exist. Its explicit value can be obtained from the condition of maximum slope at $x=a$ which leads to $\lambda_c = \pi^2/[2(a-L)^2]$. In that limit the $L^2$-norm (number of particles in the condensate) diverges, since the amplitude in the outer region $C \rightarrow \infty$. 
    \begin{figure}
 \epsfig{file=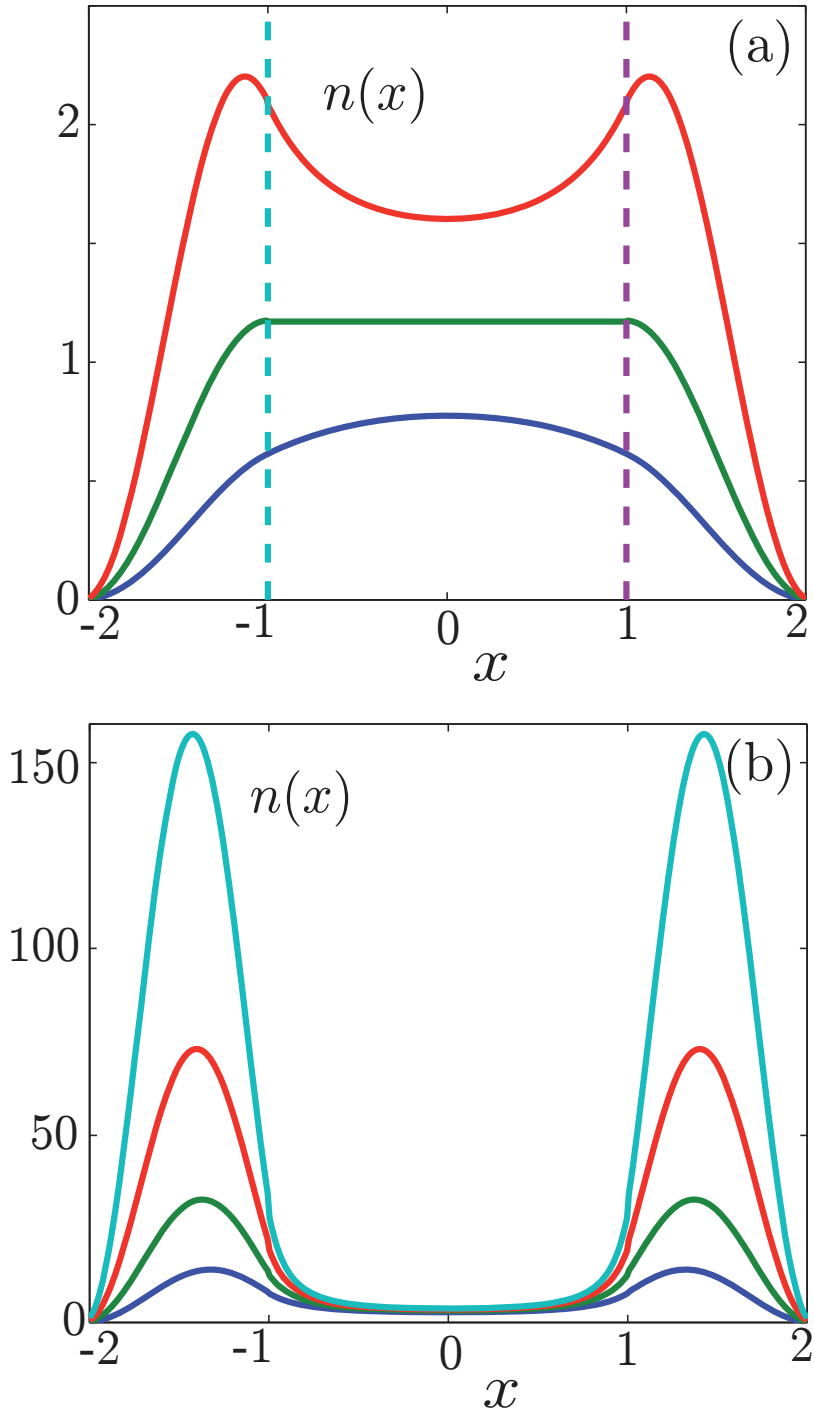,width=0.6\textwidth}
 \caption{[Color online] Spatial distribution of $n(x) = u^2(x)$ for the ground state solutions given by Eqs.  Eq. (\ref{pepa})-(\ref{pipa}) with $L=2, a =1$ and different values of $\lambda$ corresponding to different values of $N$ (a) From the lower to the upper curve $g_0N =2, 3.55$ (corresponding to $\lambda_*$) and $g_0N=6$. (b) From the lower to the upper curve $g_0N=25, 50, 100, 200$.  \label{prima}}
 \end{figure}
 
The  spatial profiles of the ground state density for different values of $g_0N$ shown in Fig. \ref{prima} support our conjecture based on the Thomas-Fermi solution, i.e. the existence of a strong localization of the solution in the region where the interactions vanish.  

Physically it is also remarkable that the atom density in the inner part of the domain, i.e. the region where there are nonlinear interactions, remains almost constant independently on the number of particles in the condensate once a certain critical density is achieved. This region is a nonlinear analogue of the classically forbidden region in ordinary potentials and is energetically less favourable due to the extra repulsive energy provided by the nonlinear interactions. However, the tunneling of atoms in this region is essentially limited to a constant value, \emph{independently of the number of atoms}, which differs essentially from ordinary tunneling. 
 
The supression of tunneling strongly depends  on the value of the scattering length in the outer region, that we will take to be nonzero in what follows. If instead we set $g(x) = g_*$ when $|x|>a$ and study the dependence of the ratio between the maximum atom density and the atom density at $x=0$ (which is a measure of the amplitude of the tunneling), we find a strong dependence on this parameter as it is shown in Fig. \ref{nueva}(a).

This effect is also seen in the atom density profiles when comparing the cases with $g_*=0$ [Fig. \ref{nueva}(b)] and $g_*=0.2$ [Fig. \ref{nueva}(c)] for $N=1000$. Larger values of $N$ lead to a stronger effect.
   
\begin{figure}
 \epsfig{file=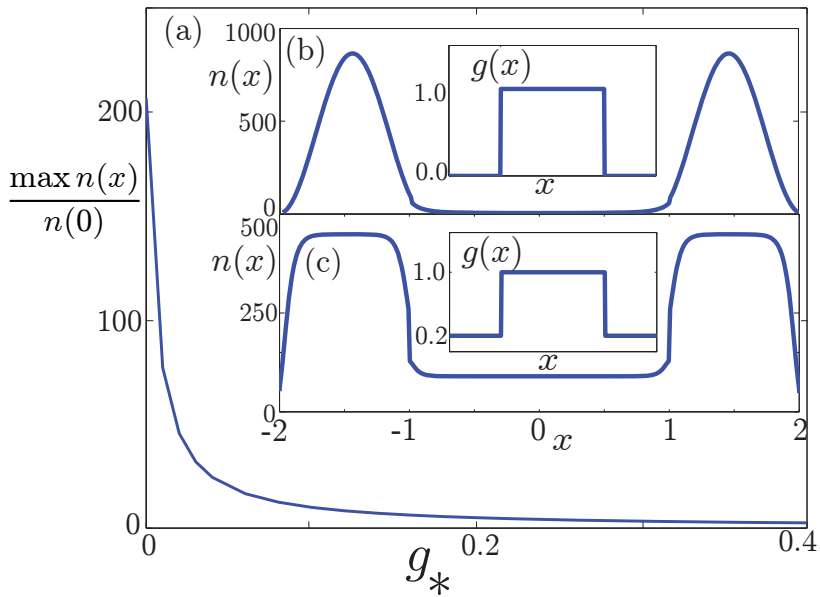,width=0.6\textwidth}
 \caption{[Color online]  (a) Ratio between the maximum atom density and the atom density at $x=0$ as a function of the scaled scattering length $g_*$ on the spatial region $|x|>a$ (for $g_0 = 1$) (b-c) Scaled atom density profiles for (b) $g_*=0$ and (c) $g_* = 0.2$ in both cases for a total scaled number of particles $N=1000$. The insets show the profile of $g(x)$. \label{nueva}}
 \end{figure}
 
Thus we have shown a simple example which already presents the main features of the nonlinear phenomena to be studied with more detail in this paper:  the localization of the ground state in the regions where the interactions vanish, the supression of tunneling induced by the presence of inhomogeneous interactions and the limited range of variation for the eigenvalue $\lambda$ (chemical potential). 
 
Related phenomena have been described in the mathematical analysis of similar logistic equations \cite{16} in bounded domains  \cite{Lopez}. However to consider situations of real physical interest we must move to unbounded domains where the analysis is much more complicated and goes beyond our particular simple example and previous theoretical studies \cite{16,Lopez}.
 
 \section{Numerical results}
 \label{numerics}
 
 \subsection{Introduction}
 
 In this section we will consider more realistic situations corresponding to problems posed on unbounded domains and strongly confining potentials satisfying the condition \eqref{hip:V}. In the examples to be presented in this section we will use harmonic potentials $V(x) \propto x^2$ which arise naturally in the modelling of magnetic trapping potentials for ultracold atoms.
 
 To study the localization phenomenon we have computed numerically the ground state solutions of equations of the form \eqref{gp_dimensionless} using a standard steepest descent algorithm on the energy functional complemented with a projection over the set of solutions of given $L^2$-norm \cite{Bao}. 
 
 We will consider nonlinear coefficients $g(x)$ also of realistic forms which vanish or are very small on certain sets. In this section we will use a more physical language thinking on the applications of our results to BEC systems.
 
 \subsection{One dimensional systems}
 
  Let us consider Eq. (\ref{gs}) in one-dimensional scenarios ($N=1$), and take the  potential to be of the form $V(x) = 0.02 x^2$.  Because of the possibilities for optical control of nonlinearities which are posible in BEC systems we will choose the  nonlinear coefficient to be of the form
  \begin{equation}
  g_a(x) = \exp\left(-x^2/2a^2\right),
  \end{equation}
  for different values of $a\leq +\infty$. These choices will allow us to study different degrees of localization of the interactions starting from the case of no localization.
  
   \begin{figure}
 \epsfig{file=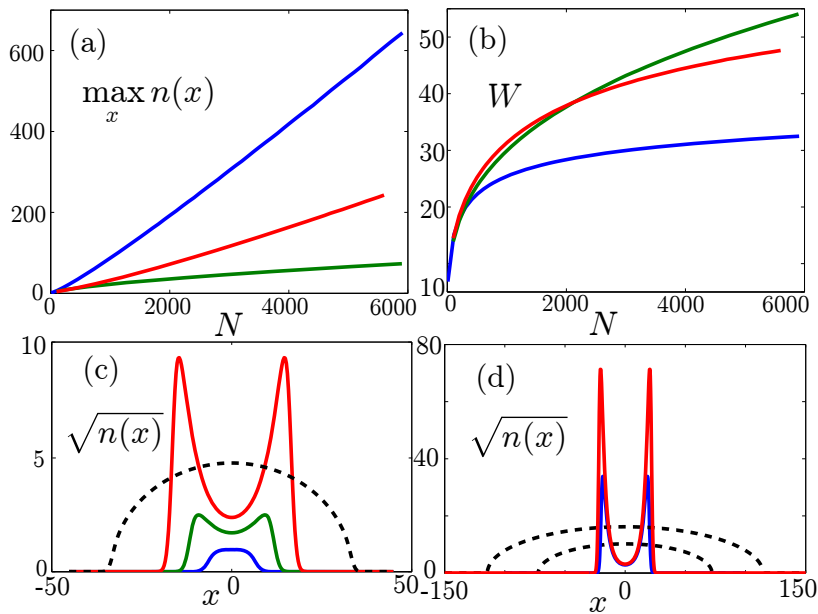,width=0.6\textwidth}
 \caption{[Color online] Ground states of Eq. (\ref{gs}) for $V(x) = 0.02 x^2$ and $g_0(x) = 1, g_1(x)=\exp (-x^2/200)$, and $g_2(x) = \exp(-x^2/50)$,  for different values of the scaled number of particles $N$. (a) Maximum particle density $\max_x n(x)$ and (b) width $W^2 = \int x^2 n(x)dx/N$ for $g_0$ (green), $g_1(x)$ (red) and $g_2(x)$ (blue). (c) Spatial profiles of $\sqrt{n(x)}$ for $N=10$ (blue), $N=100$ (green), $N=1000$ (red) for $g(x)=g_1(x)$. The dashed black line is the ground state with homogeneous interactions and $N=1000$. (d) 
Same as (c) but for $N=10000$ (blue), and $N=40000$ (green),  in comparison with the case of spatially homogeneous interactions (dashed black lines). \label{segunda}}
 \end{figure}

Our results are summarized in Fig. \ref{segunda}. First, in Fig. \ref{segunda}(a) we observe how the maximum density ($n(x) = u^2(x)$)  increases drastically for spatially decaying nonlinearities (blue and red curves) as a function of the effective number of particles in the quasi-one dimensional condensate $N$. This amplitude growth is due to a strong localization effect near the region where $g(x)$ vanishes as it is seen in Fig. \ref{segunda}(c,d). In contrast, the condensate density for spatially homogeneous interactions grows slowly according to the Thomas-Fermi  prediction $\max n(x) \propto N^{2/3}$.  When the number of particles is small, the size of the atomic cloud  is smaller than the  localization region of $g(x)$. For larger values of $N$ the ground state extends beyond the localization region of $g(x)$ and the atom density becomes more and more localized near its edge. This effect is more clear for larger $N$ and is accompanied by a saturation in the amplitude growth in the region where $g(x)$ is far from zero [Fig. \ref{segunda}(d)]. 

   \begin{figure}
 \epsfig{file=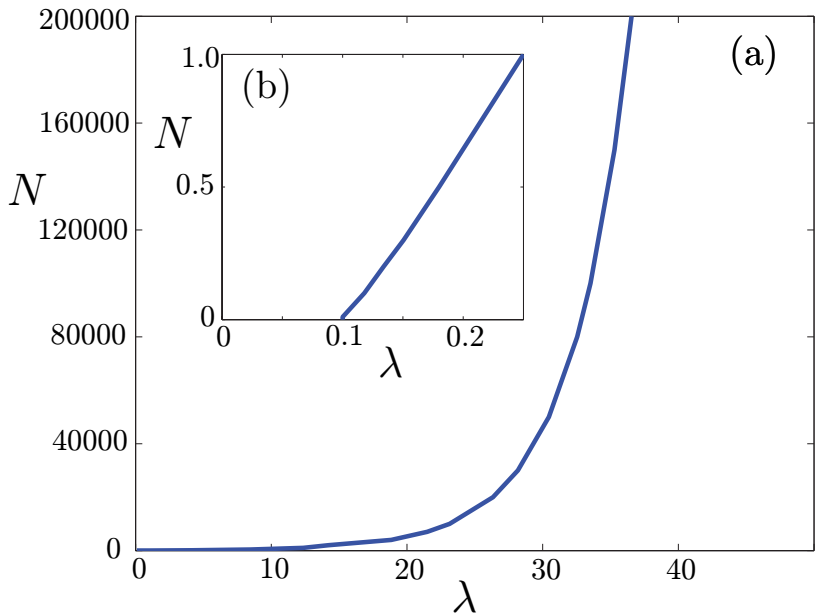,width=0.6\textwidth}
 \caption{[Color online] (a-b) Dependence of the number of particles in the ground state $N$ on the eigenvalue (chemical potential) $\lambda$
 for $g(x) = \exp(-x^2/200)$,  and $V(x) = 0.02 x^2$. Subplot (b) shows the behavior in for small $N$ where the eigenvalue approaches the linear limit.
\label{XX}}
 \end{figure}
 
In the case of spatially homogeneous interactions the width grows according to the law $W \propto N^{1/3}$  [Fig. \ref{segunda}(b)]. When interactions are spatially dependent and due to the localization of the amplitude close to the zero of $g(x)$ the width growth saturates for large values of $N$ to a value depending on the size of $g(x)$.

It is interesting to point out that although the nonlinearity does not vanish strictly anywhere there is an effect similar to the existence of the cutoff discussed in the previous example, the number of particles increasing enormously for increasing values of $\lambda$, as it is shown in Fig. \ref{XX}

These effects are even more clear when the nonlinearity decays to zero faster and is exactly zero on a certain region. For instance, taking a nonlinear coefficient given by $g_4(x) = (1-0.001 x^2)_+$ (i.e. an inverted parabola with maximum amplitude $g=1$ at $x=0$ and zero values for $|x|>31.6$) as shown in Fig. \ref{tertia}, we see that the maximum densities are even higher than for $g_1(x), g_2(x)$. In this case, in comparison with our first simple example, we can see an even stronger localization since the existence of the potential makes energetically more favourable the localization close to the point where the interactions vanish.

In this case we also observe the existence of a cutoff in the values of the eigenvalue $\lambda$ (chemical potential) close to a certain value $\lambda_c$. An example is shown in Fig. \ref{tertia2}.

  \begin{figure}
 \epsfig{file=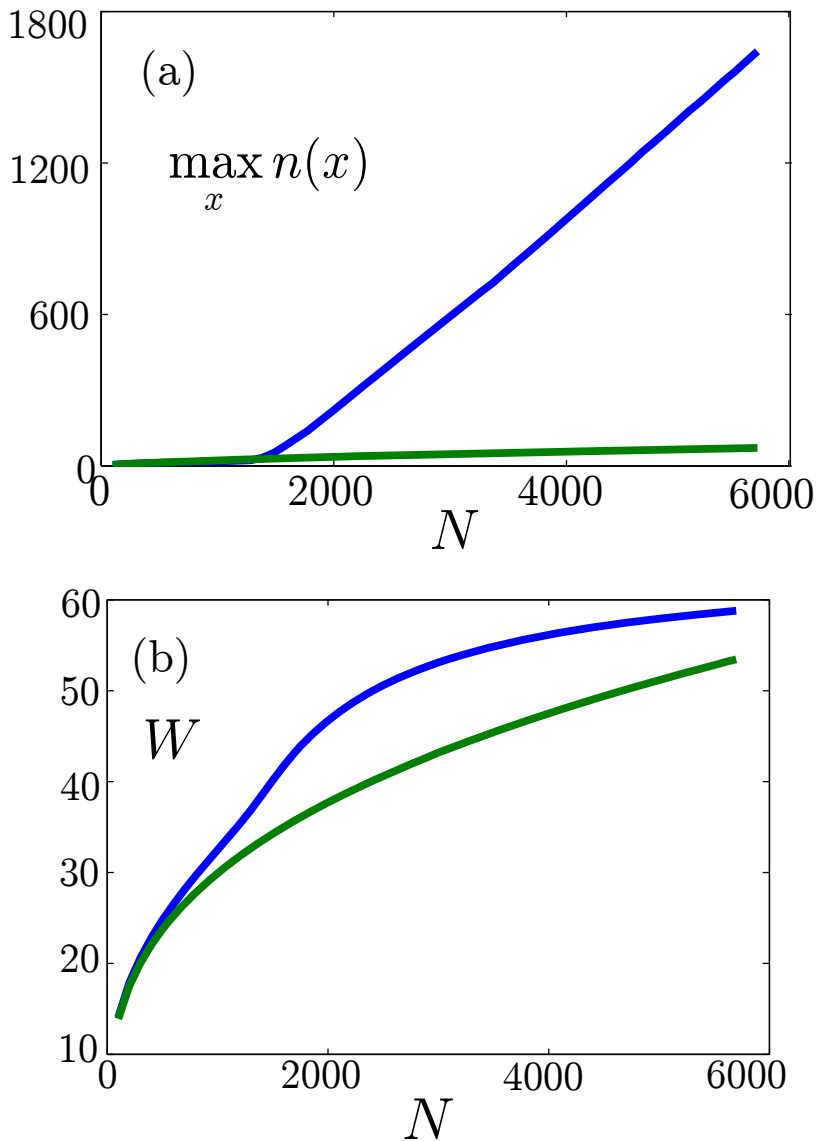,width=0.6\textwidth}
 \caption{[Color online] Ground state maximum particle density (a) and width $W = \int x^2 n(x)dx/N$ (b) for $V(x) = 0.02 x^2$ and $g_0(x) = 1$ (red lines) and $g_4(x)=(1-0.001 x^2)_+$ (blue lines). \label{tertia}}
 \end{figure}
 
   \begin{figure}
 \epsfig{file=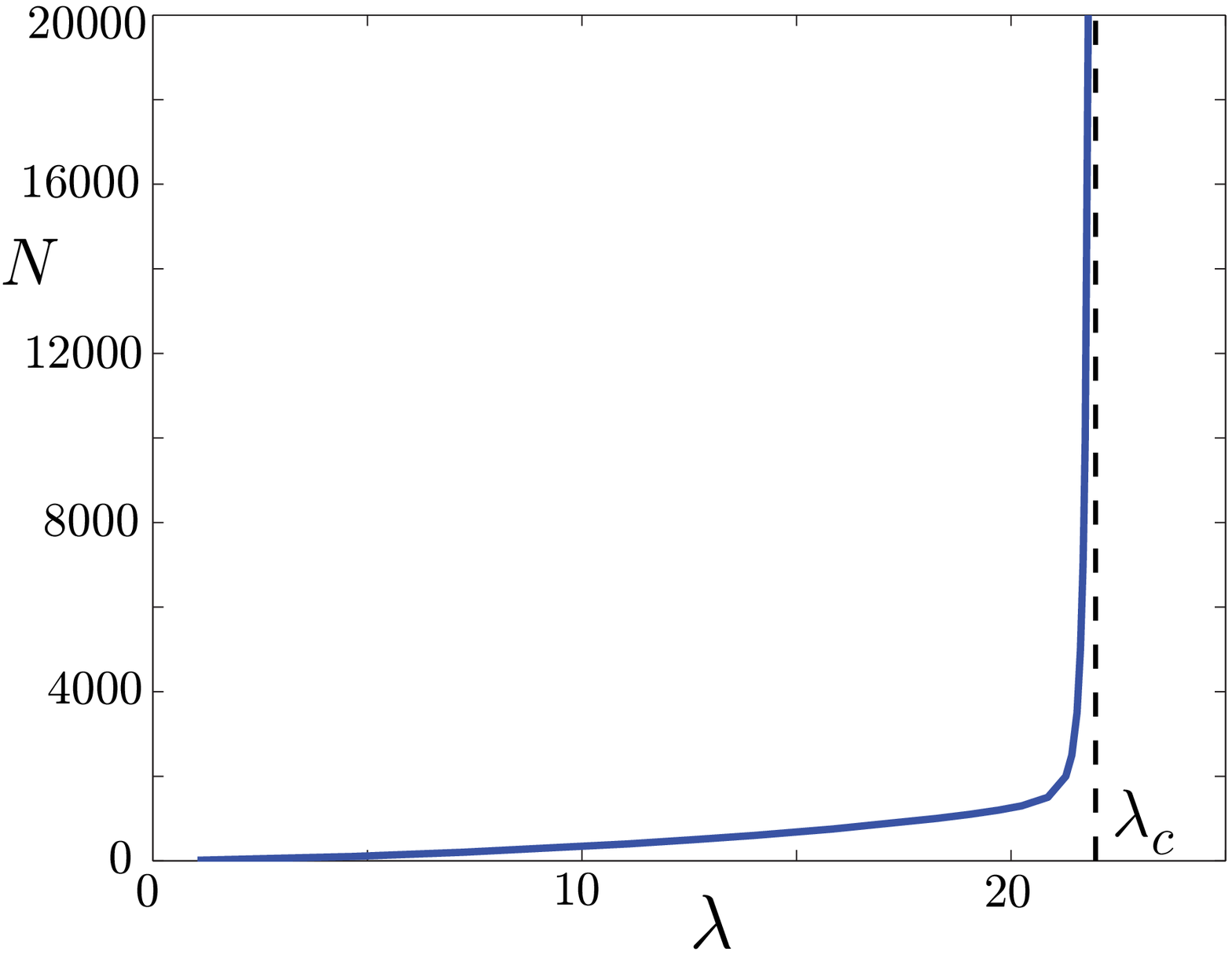,width=0.6\textwidth}
 \caption{[Color online]  Dependence of the number of particles in the ground state $N$ on the eigenvalue (chemical potential) $\lambda$
 for $g_4(x)=(1-0.001 x^2)_+$,  and $V(x) = 0.02 x^2$. \label{tertia2}}
 \end{figure}
 
As a final example, in Fig. \ref{cuarta} we show again the localization phenomenon but now near $x=0$ for localized interactions given by $g(x) = 1- \exp(-x^2/50)$ and $V(x) = 0.02 x^2$. In this case interactions vanish in a single point $x=0$ being very small in its vicinity. It is interesting  to see how localized the density becomes to ``avoid" penetrating into the regions with appreciable values of the scattering length.
    \begin{figure}
 \epsfig{file=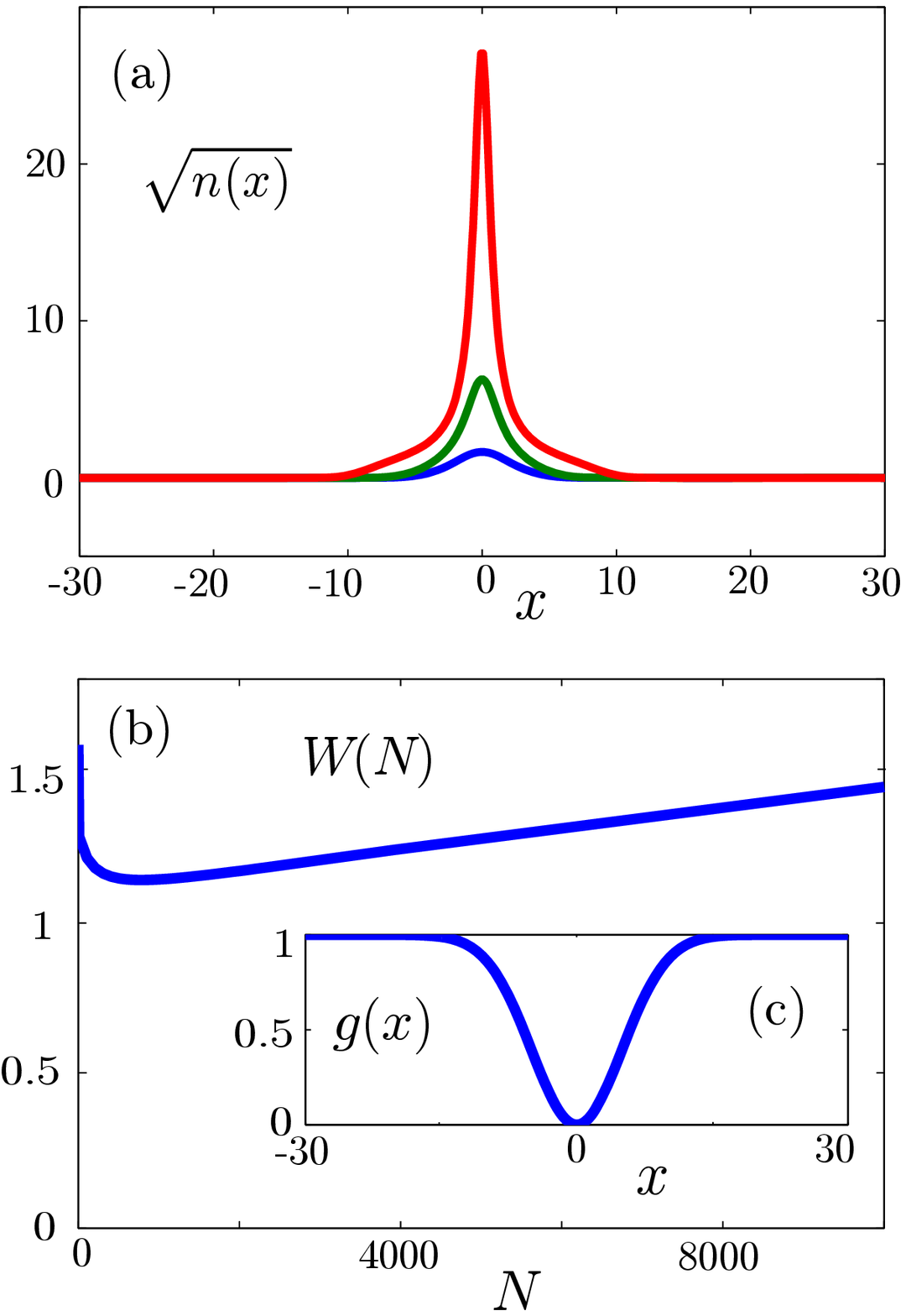,width=0.6\textwidth}
 \caption{[Color online] Ground states for $V(x) = 0.02 x^2$ and $g(x) = 1-\exp(-x^2/50)$ (shown in panel (c)). (a) $\sqrt{n(x)}$  for $N=10$ (blue line), $N=100$ (green line) and $N=1000$ (red line). (b) Condensate width as a function of $N$. \label{cuarta}}
 \end{figure}
 
\subsection{A three-dimensional example}

 The same localization phenomenon  is also present in  three-dimensional geometries. 
 
 As a simple example we will take a set of parameters of applicability to BECs describing a $^{87}$Rb condensate where the nonlinear interactions  $gN$ are typically in the range $10^3-10^5$ for typical numbers of atoms and trap sizes. As to the potential and nonlinearity we will take them to be of the form 
 \begin{subequations}
 \begin{eqnarray}
 V(x,y,z) & =& \frac{1}{2}\left(x^2+y^2+z^2\right), \\
 g(x,y,z)  & = & g_0\left[1-\exp\left(-\frac{x^2+y^2+z^2}{2w^2}\right)\right],
 \end{eqnarray}
\end{subequations}
 with $w = 0.5$ (i.e. half the trap width). Our results are summarized in Fig. \ref{3D} where again the same localization phenomenon, here inside the ``hole" of the nonlinear coefficient, is seen. For larger number of atoms the localization phenomenon can be even enhanced as in the previous one dimensional examples.

  \begin{figure}
 \epsfig{file=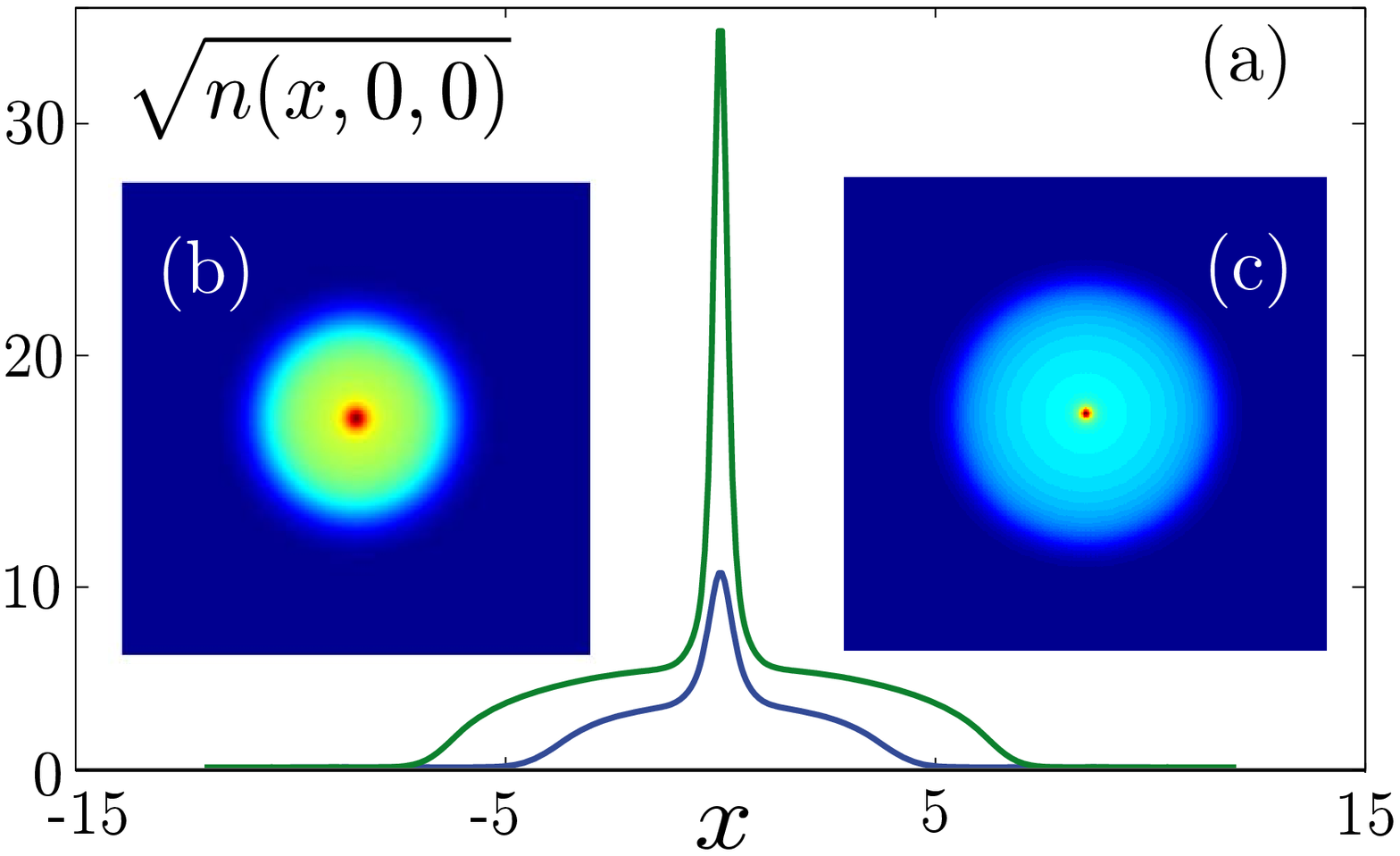,width=0.6\textwidth}
\caption{[Color online] (a) Comparison of the profile of $\sqrt{n(x,0,0)}$ for $g_0N = 10^3$ (blue line) and $g_0N= 10^4$ (red line) showing the localization of the atom density. (b-c) Pseudocolor plot of $\sqrt{n(x,y,0)}$ for (b) $g_0N =  10^3$ on the spatial region $(x,y) \in [-10,10]\times [-10,10]$ and 
(c) $g_0N =  10^4$ on the spatial region $(x,y) \in [-15,15]\times [-15,15]$ \label{3D}}
\end{figure}

 \section{Theory} 

\subsection{Introduction}

In Sections \ref{toy} and \ref{numerics} we have presented an analytic simple example and several numerical simulations showing a striking localization of the solutions of Eq. (\ref{eq:rd}) close to the zeros of the nonlinear coefficient $g(x)$. Here, we will provide a rigorous support for two of the observed phenomena: (i) the existence of a finite range of values of the chemical potential (i.e. the existence of a cutoff) and (ii) the divergence of the amplitude of the solution close to that point.

In this section we will use a more formal and precise language and  provide some mathematically rigorous results which will support our results in general scenarios. 

\subsection{Preliminaries and notation}
\label{preliminaries:notation}

We shall fix the potential $V$ satisfying hypothesis (\ref{hip:V}). Following \cite{welsh},
let $\Omega \subset \mathbb{R}^N  $ be an open nonempty set, possibly unbounded, with boundary regular enough, let us denote by
\begin{equation}\label{def:H:Omega}
H(\Omega ,V):=\left\{ u\in H^1_0( \Omega)\ : \ \int_{ \Omega} V\, u^2 < +\infty \right\},
\end{equation}
$H(\Omega ,V)$ is the completion of $C_0^\infty (\Omega)$ in the metric derived from the norm
\begin{equation}\label{norm:H:Omega}
\|u\|:=\left(\int_{ \Omega} |\nabla u|^2 + V\, u^2  \right)^{1/2},
\end{equation}
and $H(\Omega ,V)$  is  a Hilbert space with the scalar product
\begin{equation}\label{escalar:product:H:Omega}
\langle u, v\rangle :=\int_{ \Omega} \nabla u \nabla v + V\, uv.
\end{equation}

Let us consider the operator $L$ defined by 
\begin{equation}\label{def:op:L:Omega}
Lu:= -\Delta  u + V\, u, \qquad \mbox{for } u\in D(L,\Omega )
\end{equation}
where
\begin{equation}\label{def:dom:L:Omega}
D(L,\Omega ):=H(\Omega ,V)\cap
\{ u\in L^2(\Omega )\ :\ -\Delta  u + V\, u \ \in L^2( \Omega) \}.
\end{equation}
Let us first write the following lemma

\begin{lemma}\label{lemma:1}
Let $\Omega \subset \mathbb{R}^N  $ be an open nonempty set, possibly unbounded, with boundary regular enough. If  $V$ satisfy
hypothesis {\rm (\ref{hip:V})}, then the following assertions are true

\begin{description}
\item[{\rm i)}] $H(\Omega ,V)\hookrightarrow L^2( \Omega)$ with compact embedding

\item[{\rm ii)}] the operator $L$ defined by Eq. (\ref{def:op:L:Omega}), has a discrete spectrum, noted by $\sigma (L, \Omega ),$
i.e. $\sigma (L, \Omega )$ consists of an infinite sequence of isolated eigenvalues $\{ \sigma _n ( \Omega)\}\uparrow\infty $
with finite multiplicities.

\item[{\rm iii)}] Moreover, the Rayleigh sup-inf characterization for the eigenvalues holds, and in particular the first
eigenvalue, denoted by $\ \sigma _1( \Omega ),\ $ satisfies
\begin{equation}\label{cociente:Rayleigh}
\sigma _1( \Omega )= \inf_{\psi\in H( \Omega,V)} \ \frac{\displaystyle\int_{ \Omega} |\nabla \psi|^2
+ V\, \psi^2}{\displaystyle\int_{ \Omega} \, \psi^2}
\end{equation}

\item[{\rm iv)}] the first eigenvalue  is positive, simple with a positive eigenfunction, $\ \phi_1( \Omega )>0,$ and there is no
other eigenvalue with a positive eigenfunction.
\end{description}
\end{lemma}

This lemma consists of known results, being folk for $N=1,2.$ In higher dimensions, parts i)-iii) are contained in \cite{welsh} for $N>2,$
more precisely, see \cite[Theorem 1.2]{welsh} and the beginning of its proof for parts i) and ii) respectively, see \cite[Lemma
2.1]{welsh} for part iii). With respect to part iv), it is also well known that the elliptic operator $L$ as defined in
(\ref{def:op:L:Omega}) admits a unique principal eigenvalue in $\Omega.$ This principal eigenvalue is the bottom of the spectrum
of $L$ in the  functional space $H(\Omega,V),$ and it admits an associated positive principal eigenfunction, denoted by $\phi_1>0.$
This result follows from the Krein--Rutman theory and from compactness arguments (cf. \cite{Krein-Rutman} and
\cite{Krasnoselskij-Lifshits-Sobolev}).

In general, the domain of the operator, the eigenvalues and the eigenfunctions depend not only on the domain $ \Omega$ but also
on the potential $V$ i.e.   $D(L, \Omega )=D(L,\Omega ,V),$ $\sigma _n  ( \Omega)= \sigma _n  (\Omega,V),$ $\phi_n ( \Omega)=
\phi_n (\Omega,V)$ and so on. We will consider a fixed $V$ and in what follows we will skip the dependence on $V$ to
clarify the notation, except for the ambiguous cases.

The following lemma is the Lax-Milgram lemma, see \cite[Lemma V.8]{brezis}. Let $\left[H(\Omega ,V)\right]^*$ denote the dual
space of all linear and continuous functionals defined on $ H(\Omega ,V),$

\begin{lemma}
\label{lemma:lax:milgram}
For any $f\in \left[H(\Omega ,V)\right]^*$ there exists a unique $u\in H(\Omega ,V)$ such that
\begin{equation}\label{eq:weak:sol}
\int_{ \Omega} \nabla u \nabla \psi + V\, u\psi= \int_{ \Omega} f\psi, \qquad \forall \psi\in H(\Omega ,V).
\end{equation}
\end{lemma}

\begin{remark}
\label{compactness} The above Lemma can be understood in the following way, the inverse operator $ L^{-1}:\left[H(\Omega
,V)\right]^*\to H(\Omega ,V)$ is well defined. Thanks to the compact embedding in lemma \ref{lemma:1}.i), we can consider the
inverse operator $ L^{-1}$ restricted to $L^2(\Omega )$ as a compact operator from $L^2(\Omega )$ into itself.
\end{remark}

Next, we will compare the eigenvalues defined on $\Omega$ with the eigenvalues defined in the whole $\mathbb{R}^N.$ We will
denote by $H,$ $D(L),$ $\sigma _n  ,$ $\phi_n , \cdots $ the space, the domain of the operator, the eigenvalues and the
eigenfunctions and so on whenever the operator $L$ is defined on the whole $\mathbb{R}^N.$   More precisely, let us define the
Hilbert space
\begin{equation}\label{def:H}
H:=\left\{ u\in H^1( \mathbb{R}^N) : \int_{ \mathbb{R}^N} V\, u^2 < +\infty \right\},
\end{equation}
and consider the operator $L$ defined by
\begin{equation}\label{def:op:L}
Lu:= -\Delta  u + V\, u, \qquad \mbox{for } u\in D(L)
\end{equation}
where now
\begin{equation}\label{def:dom:L}
D(L):=H\cap
\left\{ u\in L^2( \mathbb{R}^N) : -\Delta  u + V\, u \in L^2( \mathbb{R}^N) \right\}.
\end{equation}
Observe that the above lemmas \ref{lemma:1}, \ref{lemma:lax:milgram} still hold, in particular $H\hookrightarrow L^2( \mathbb{R}^N)$,
with compact embedding whenever $N>2$ and $V$ satisfies hypothesis (\ref{hip:V}). The elliptic operator $L$ as defined in (\ref{def:op:L}) admits a unique
principal eigenvalue  in $\mathbb{R}^N$ noted by $\sigma_1 .$ This principal eigenvalue is the bottom of the spectrum of L in the
function space $H,$ and it admits an associated positive principal eigenfunction.

In the following lemma we collect  the monotonicity properties of the eigenvalues with respect to the domain, see \cite[Theorem
2.3]{welsh} for a proof. As a consequence, we can compare the eigenvalues defined on $\Omega \subsetneq \mathbb{R}^N $ with the
eigenvalues defined in the whole $\mathbb{R}^N.$

\begin{lemma}\label{lemma:2}
Let $\sigma _1(\Omega )\leq \sigma _2(\Omega )\leq \cdots  $ be the eigenvalues of $L,$ with corresponding eigenfunctions $\phi
_1(\Omega ),\ \phi _2(\Omega ) \cdots $ defined on $H(\Omega ).$ For a subdomain $\Omega ^*\subset \Omega$ with boundary regular
enough, let $\sigma _1(\Omega ^*)\leq \sigma _2(\Omega ^*)\leq \cdots  $ be the eigenvalues of $L,$ with corresponding
eigenfuctions $\phi _1(\Omega ^*),\ \phi _2(\Omega ^*) \cdots  $ defined on $H(\Omega ^*),$ then
\begin{equation}\label{monotonicity:sigma}
\sigma _n(\Omega ^*) > \sigma _n(\Omega ).
\end{equation}
\end{lemma}

In particular, let $\sigma _1\leq \sigma _2\leq \cdots  $ be the eigenvalues of $L,$ with corresponding eigenvectors $\phi _n $
defined on $H,$ then
\begin{equation}\label{monotonicity:sigma:RN}
\sigma _n(\Omega ) > \sigma _n.
\end{equation}

The following lemma states the Maximum Principle for unbounded domains, see \cite[Theorem IX.27]{brezis}.

\begin{lemma}\label{lemma:PM} {\bf The maximum principle for the Dirichlet problem.}
Let  $\Omega \subset \mathbb{R}^N  $ be an open nonempty set, possibly unbounded, with boundary of class $C^1,$  and assume $V$
satisfy hypothesis {\rm (\ref{hip:V})}. Let $f\in L^2( \Omega)$ and $u\in H(\Omega ,V)$ be such that (\ref{eq:weak:sol}) holds.

Then
\begin{equation}\label{eq:weak:sol:down:up}
\min \{\inf_{\partial\Omega } u, \inf_{\Omega } f\} \leq   u\leq
\max \{\sup_{\partial\Omega } u, \sup_{\Omega } f \}
\end{equation}
where $\sup=\sup \mbox{\rm ess} $ and $\inf=\inf\mbox{\rm ess}$.

In particular, if
\begin{equation}\label{eq:weak:ineq}
u\geq 0 \mbox{ on } \partial \Omega, \quad \mbox{and} \quad f\geq 0 \mbox{ in }  \Omega,
\end{equation}
then
\begin{equation}\label{eq:weak:sol:ineq}
u\geq 0 \mbox{ in }  \Omega, \quad \mbox{and} \quad \|u\|_{L^{\infty}(\Omega )}\leq
\max \{\|u\|_{L^{\infty}(\partial\Omega )}, \|f\|_{L^{\infty}(\Omega )} \}.
\end{equation}
\end{lemma}

\subsection{Main result}

We now consider  Eq. (\ref{eq:rd}) as a bifurcation problem, considering $\lambda$ as a real parameter, we look for pairs
$(\lambda ,u_{\lambda})\in \mathbb{R}\times H $ such that  $u_{\lambda}$ is a positive solution of (\ref{eq:rd}).

As in Sec. \ref{preliminaries:notation}, let $\sigma_1$ stand for the first eigenvalue of the eigenvalue problem
\begin{equation}
\label{def:sigma1:RN} \left( -\frac{1}{2}\Delta   + V(x)  \right)\phi _1:= \sigma_1\phi _1, \quad x\in\mathbb{R}^N , \quad \phi
_1(x)\to 0, \mbox{ as } |x|\to \infty ,
\end{equation}
and given an open regular enough  domain $\Omega\subset \mathbb{R}^N,$ let $\sigma_1=\sigma_1( \Omega)$ stand for the first
eigenvalue of the Dirichlet eigenvalue problem
\begin{equation}
\label{def:sigma1:Omega} \left( -\frac{1}{2}\Delta   + V(x)  \right)\phi _1:= \sigma_1\phi _1, \quad x\in\Omega , \quad \phi
_1(x)=0, \mbox{ on }\partial \Omega ,
\end{equation}
where the first eigenfunction $\phi_1=\phi_1( \Omega)$.

Let  $\omega $ the set where $g$ vanishes, defined by Eq. (\ref{set:supp:g}). 
Technically we will require that this coefficient is bounded away from zero asymptotically, i. e.
\begin{equation}
g(x)\geq g_0>0, \qquad \mbox{for}  \ |x|\gg 1, \label{hip:g}
\end{equation}
(note that in our numerical simulations we have observed the phenomenon described in the paper in less strict situations). 
In general, $\omega $ is not a connected set, let
$\sigma_1 ( \omega )$ be the minimum
\begin{equation}
\label{def:sigma1:omega} \sigma_1 ( \omega ):= \min \{\sigma_1( \omega_j): 1\leq j\leq J \},
\end{equation}
where $\omega_j: 1\leq j\leq J$ are defined in (H2).

\begin{theorem}\label{teoteo}
The problem (\ref{eq:rd})   has a positive solution $(\lambda ,u_{\lambda})$ whenever
\begin{equation}
\label{ineq:lambda1} \sigma_1 < \lambda < \sigma_1 ( \omega ),
\end{equation}
moreover
\begin{eqnarray}
\label{L:infty:norm} \|u_{\lambda}\|_{L^{\infty}(\mathbb{R}^N)} \to 0,  \mbox{ as }  \lambda \downarrow \sigma_1 , \\
\|u_{\lambda}\|_{L^{\infty}(\mathbb{R}^N)} \to \infty ,  \mbox{ as }  \lambda \uparrow  \sigma_1 ( \omega ).
\end{eqnarray}
\end{theorem}

\begin{proof}
The Crandall--Rabinowitz's bifurcation theorem, \cite{C-R-73} implies that $(\sigma_1 ,0)$ is a bifurcation point in
$\mathbb{R}\times H.$
Moreover, the bifurcation is supercritical, i.e. on the right of the eigenvalue.

And, the monotonicity with respect to the parameter implies the uniqueness.

\begin{equation}\label{eq:lambda}
\lambda = \sigma_1 (\mathbb{R}^N,  V+g u^{p-1} )< \sigma_1(\omega ,  V+g u^{p-1} )=\sigma_1(\omega ,  V ).
\end{equation}

\end{proof}

 \section{Conclusions and discussion} 
 
 We have studied the ground state of nonlinear Schr\"odinger equations with confining potentials and repulsive interactions vanishing on nonempty sets $\omega$ and shown that  there is a nonlinear localization phenomenon of the solution on this set. We have presented analytical formulae and numerical simulations showing it and provided a theorem which rigorously explains some of the observed features. 
 
 The results have interesting implications to the physics of  Bose-Einstein condensates with repulsive spatially varying interactions because in that situation the atom density will localize dramatically near the regions where the interactions are close to zero when the nonlinear interactions are relevant enough. This behavior does not depend on the spatial dimensionality and can be used to design very effectively one, two or three-dimensional spatial distributions with large values of the atom density by acting on the control field, e.g.
 using micromagnets to change locally the magnetic field or localized laser beams. 
 
 The behavior described in this paper cannot be achieved with ordinary potentials since due to tunneling of the wavefunction in linear potentials it is not posinble to achieve such a high degree of localization. However, the nonlinear interactions here are essential in limiting the tunneling to the regions with larger values of the scattering length which is independent on the number of particles. 
 
 Our results could be useful in different applications of Bose-Einstein condensates such as for atom lithography, atom beam guiding  or other applications where a precise control of the positioning of large values of the density of an atomic cloud is required.

These phenomena  may appear in multicomponent condensates that offer a much wider range of possibilities for their ground states depending on the type of components and their relative interactions \cite{Last}. It would be interesting to study the situation in which the variations of the scattering length make the components to be miscible or inmiscible in different spatial regions and  what is the geometry of the resulting domains.

Finally, it is interesting that the phenomena reported here could also be observed in other systems ruled by nonlinear Schr\"odinger equations similar to Eq. (\ref{gs}), for instance in nonlinear optical systems. Although ordinary optical materials have small nonlinearities were these phenomena could not be easily observed except for unrealistically large laser intensities, there are media with a giant nonlinear response \cite{EIT1,EIT2} as it happens when a probe laser beam propagates in a medium with  transparency induced electromagnetically by a second coupling field. Therefore, managing  the parameters it may be posible to find optical version of the phenomena presented here.

\textbf{Acknowledgements.} This work has been partially supported by grants FIS2006-04190 and MTM2006-08262 (Ministerio de Educaci\'on y Ciencia, Spain) and PCI-08-0093 (Consejer\'{\i}a de Educaci\'on y Ciencia, Junta de Comunidades de Castilla-La Mancha, Spain).

\end{document}